\documentclass[smallabstract,smallcaptions]{dccpaper}
\pdfoutput=1

\usepackage[left=1.25in, top=1in]{geometry}
\usepackage{amsmath}
\usepackage{amsthm}
\usepackage{setspace}
\usepackage[hidelinks]{hyperref}
\usepackage{microtype}
\usepackage{graphicx}
\usepackage[capitalize]{cleveref}
\usepackage[inline]{enumitem}
\usepackage{tikz}
\usepackage[bf]{caption}
\usepackage[ruled,noend,linesnumbered]{algorithm2e}

\usetikzlibrary{decorations.pathreplacing}

\newtheorem{theorem}{Theorem}[section]
\newtheorem{lemma}{Lemma}[section]
\newtheorem{corollary}{Corollary}[section]

\newtheorem{conjecture}{Conjecture}[section]

\theoremstyle{definition}
\newtheorem{definition}{Definition}[section]
\theoremstyle{remark}
\newtheorem{remark}{Remark}[section]

\setlist[enumerate]{nosep, topsep=1ex}
\setlist[itemize]{nosep, topsep=1ex}
\setlist[description]{nosep,topsep=1ex}

\allowdisplaybreaks

\SetCommentSty{textit}
\SetKwInOut{Input}{Input}
\SetKwInOut{Output}{Output}

\makeatletter
\patchcmd\algocf@Vline{\vrule}{\vrule \kern-0.4pt}{}{}
\patchcmd\algocf@Vsline{\vrule}{\vrule \kern-0.4pt}{}{}
\makeatother

\newcommand{\emptystring}{\varepsilon}

\newcommand{\dol}{{\rm \$}}
\newcommand{\hash}{\text{\#}}

\newcommand{\rhs}[1]{{\rm rhs}(#1)}
\newcommand{\rhsgen}[2]{{\rm rhs}_{#1}(#2)}
\renewcommand{\exp}[1]{{\rm exp}(#1)}
\newcommand{\expgen}[2]{{\rm exp}_{#1}(#2)}

\newcommand{\len}[1]{{\rm len}(#1)}
\newcommand{\lengen}[2]{{\rm len}_{#1}(#2)}

\newcommand{\bigO}{\mathcal{O}}
\newcommand{\dd}{\mathinner{.\,.}}

\newcommand{\probname}[1]{\text{\sc #1}}

\Support{This research has been supported by the NSF CAREER Award
  2337891 and the Simons Foundation Junior Faculty Fellowship.}
\title{\textbf{Word Break on SLP-Compressed Texts}}
\author{Rajat De and Dominik Kempa\\[1.5ex]
  \small Department of Computer Science,\\[-0.2ex]
  \small Stony Brook University, Stony Brook, NY, USA\\[-0.2ex]
  \small \url{rde@cs.stonybrook.edu}, \url{kempa@cs.stonybrook.edu}\\
}

\date{}

\begin{document}

\maketitle
\vspace{3ex}

\begin{abstract}
  Word Break is a prototypical factorization problem in string
  processing: Given a word $w$ of length $N$ and a dictionary
  $\mathcal{D} = \{d_1, d_2, \ldots, d_{K}\}$ of $K$ strings,
  determine whether we can partition $w$ into words from
  $\mathcal{D}$.  We propose the first algorithm that solves the Word
  Break problem over the SLP-compressed input text $w$. Specifically,
  we show that, given the string $w$ represented using an SLP of size
  $g$, we can solve the Word Break problem in $\bigO(g \cdot
  m^{\omega} + M)$ time, where $m = \max_{i=1}^{K} |d_i|$, $M =
  \sum_{i=1}^{K} |d_i|$, and $\omega \geq 2$ is the matrix
  multiplication exponent. We obtain our algorithm as a simple
  corollary of a more general result: We show that in $\bigO(g \cdot
  m^{\omega} + M)$ time, we can \emph{index} the input text $w$ so
  that solving the Word Break problem for any of its substrings takes
  $\bigO(m^2 \log N)$ time (independent of the substring length). Our
  second contribution is a lower bound: We prove that, unless the
  Combinatorial $k$-Clique Conjecture fails, there is no combinatorial
  algorithm for Word Break on SLP-compressed strings running in
  $\bigO(g \cdot m^{2-\epsilon} + M)$ time for any $\epsilon > 0$.
\end{abstract}

\section{Introduction}\label{sec:intro}

Given a string $w \in \Sigma^{N}$ and a dictionary $\mathcal{D} =
\{d_1, d_2, \ldots, d_K\}$ of $K$ strings over alphabet $\Sigma$, the
\emph{Word Break} problem asks whether there exists a
\emph{factorization} (also called a \emph{parsing}) $w = f_1 f_2
\cdots f_{\ell}$ such that for every $i \in [1 \dd \ell]$, it holds
$f_i \in \mathcal{D}$.

The Word Break is a prototypical factorization problem, where the
dictionary is an arbitrary set specified as input.  Specializations of
this problem, where the dictionary is a set (usually infinite) of
words satisfying some specific conditions, constitute many central
problems in string processing. This includes the \emph{Lyndon
  factorization}~\cite{CFL58}, \emph{square
  factorization}~\cite{MatsuokaIBTM16}, \emph{palindromic
  factorizations}~\cite{BannaiGIKKPS18,FiciGKK14}, or the \emph{closed
  factorization}~\cite{BBGIIIPS16}. In data compression, on the other
hand, the dictionary is often computed dynamically during the parsing
process, e.g., in the \emph{LZ77}~\cite{LZ77},
\emph{LZ78}~\cite{LZ78}, \emph{LZW}~\cite{LZW}, or \emph{LZ-End
  factorization}~\cite{kreft2013compressing}, the \emph{prefix-free
  parsing}~\cite{BoucherGKLMM19}, or \emph{lexicographical
  parsing}~\cite{NavarroOP21}.  Nearly all of the above factorizations
can be computed in linear time
(e.g.,~\cite{MatsuokaIBTM16,BorozdinKRS17,BBGIIIPS16,KarkkainenKP13,KempaK17,BoucherGKLMM19,NavarroOP21}).
The Word Break problem, being a more
general factorization, appears to require more time: The fastest
algorithm runs in $\widetilde{\bigO}(M + |w| \cdot M^{1/3})$
time, where $M = \sum_{i=1}^{K}|d_i|$~\cite{fastwordbreak}.\footnote{The
$\widetilde{\bigO}(\cdot)$ notation hides factors polylogarithmic in
the input length $|w|$.}

The above algorithms are usually sufficient if the input text fits
into the main memory.  However, many moderns textual datasets (such as
DNA datasets or source code repositories) satisfy only the weaker
condition: their \emph{compressed representation} fits into RAM.  This
prompted researchers to seek algorithms to operate directly on the
input text in the compressed form. One of the most common compression
frameworks, employed in many previous studies on compressed algorithms
due to its generality and relative ease of use, is \emph{grammar
compression}~\cite{CharikarLLPPSS05}. In this method, we represent
the input text $w$ as a \emph{straight-line program (SLP)}: a
context-free grammar in Chomsky-normal form encoding only the input
text.
Algorithms operating on SLP-compressed
text are known for many of the above factorization problems, including
the Lyndon factorization~\cite{FuruyaNIIBT18}, the LZ77
factorization~\cite{KempaK23}, or the LZ78
factorization~\cite{BannaiGIT13}.  No prior studies, however,
addressed the general Word Break problem.  We thus ask:
\[
  \text{\emph{Can the Word Break problem be efficiently solved on an
      SLP-compressed text?}}
\]

We answer the above question positively and present an algorithm that,
given the SLP $G$ of $w$ and the dictionary $\mathcal{D}$, solves the
Word Break problem in $\bigO(|G| \cdot m^{\omega} + M)$ time, where
$\omega$ is the matrix multiplication exponent, $m = \max_{i=1}^{K}
|d_i|$, and $M$ as before is the total length of strings in
$\mathcal{D}$.  Unless the dictionary size dominates the text length,
this new algorithm is capable of up to exponential speedup compared to
algorithms running on uncompressed data.
We obtain our algorithm as a corollary of
a more general result: We show that in $\bigO(|G| \cdot m^{\omega} +
M)$ time we can \emph{index} the string $w$, so that Word Break on any
substring of $w$ can be solved in $\bigO(m^2 \log N)$ time,
independent of the substring length (\cref{th:index}).

Our second contribution is a lower bound: We prove that unless
the Combinatorial $k$-Clique Conjecture fails (see
\cref{sec:lower-bound}), there is no combinatorial algorithm for the
Word Break problem that, given an SLP $G$ encoding the input text $w$,
runs in $\bigO(|G| \cdot m^{2-\epsilon} + M)$ time, for any
$\epsilon>0$ (\cref{th:lower-bound}).

\section{Preliminaries}\label{sec:prelim}

\paragraph{Strings}

Let $w \in \Sigma^N$ be a \emph{string} of length $N$ over
\emph{alphabet} $\Sigma$. For simplicity, we assume $|\Sigma| =
\bigO(1)$. Strings are indexed from $1$ to $N$. For $1 \leq i \leq j
\leq N$, we denote substrings of $w$ as $w[i \dd j]$. We denote the
length of string $w$ as $|w|$. The concatenation of strings $u$ and
$v$ is written as $u\cdot v$ or $uv$, and the empty string is denoted
$\emptystring$.

Let $M \in \{0,1\}^{N \times N}$ be a \emph{matrix} of dimension $N
\times N$. Matrix rows and columns are indexed from $0$ to $N-1$, the
entry at the intersection of row $i$ and column $j$ is denoted by
$M[i,j]$. All matrices involved in this paper have boolean entries in
the $(\lor,\land)$ semi-ring, i.e., $(A \times B)[i,j] = \lor_{k}
A[i,k] \land B[k,j]$. Transpose of a vector $v$ is denoted by
$v^T$. The constant $\omega$ is such that the current best algorithm
for multiplying two $N \times N$ matrices takes $\bigO(N^{\omega})$
time. Recently $\omega < 2.371552$ was shown in \cite{currentomega}.

\vspace{-1.5ex}
\paragraph{Grammars}

A \emph{context-free grammar} (CFG) is a tuple $G \,{=}\, (V, \Sigma,
R, S)$, where:
\begin{itemize}
  \item $V$ is the finite set of \emph{nonterminals},
  \item $\Sigma$ is a finite set of \emph{terminal} symbols,
  \item $R \subseteq V \times (V \cup \Sigma)^*$ is a set of
    \emph{productions} or \emph{rules}, and
  \item $S \in V$ is the special \emph{starting nonterminal}.
\end{itemize}

We say that $u \in (V \cup \Sigma)^{+}$ \emph{derives} $v$, and write
$u \Rightarrow^* v$, if $v$ can be obtained from $u$ by repeatedly
replacing nonterminals according to the rule set $R$. The
\emph{language} of grammar $G$ is the set $L(G) := \{w \in \Sigma^*
\mid S \Rightarrow^* w\}$.

By a \emph{straight-line grammar (SLG)} we mean an CFG $G =
(V,\Sigma,R,S)$ satisfying $|L(G)| = 1$ and in which every nonterminal
occurs in some $\gamma$ derived from $S$. Note that in an SLG, for
every $v \in V$, there exists at most one production $(v,\gamma) \in
R$ with $v$ on the left side. The string $\gamma$ is called the
\emph{definition} of nonterminal $v$, and is denoted
$\rhsgen{G}{v}$. If $G$ is clear from the context, we simply write
$\rhs{v}$. Note also that for every $\alpha \in (V \cup \Sigma)^{*}$,
there exists exactly one string $\gamma \in \Sigma^{*}$ satisfying
$\alpha \Rightarrow^{*} \gamma$. Such $\gamma$ is called the
\emph{expansion} of $\alpha$, and is denoted $\expgen{G}{\alpha}$, or
$\exp{\alpha}$, if $G$ is clear. We also denote $\lengen{G}{\alpha} =
|\expgen{G}{\alpha}|$, or simply $\len{\alpha} = |\exp{\alpha}|$. We
define the size of an SLG as $|G| = \sum_{v \in V}|\rhs{v}|$.

An SLG in which for every $v \in V$, it holds $\rhs{v} = uw$, where
$u,w \in V$, or $\rhs{v} = a$, where $a \in \Sigma$, is called a
\emph{straight-line program (SLP)}. Every SLG
$G$ can be easily transformed into an
SLP $G'$, where $|G'| = \Theta(|G|)$ and $L(G') = L(G)$.

\section{Uncompressed Word Break}\label{sec:uncompressed}

\paragraph{Prior Work}

\probname{Word Break} is a textbook problem with many folklore dynamic
programming solutions. These solutions typically involve computing an array
with information about partitioning every prefix or suffix into words from
the dictionary. Given the input string $w$ of length $|w| = N$ and a dictionary
$\mathcal{D}$ of size $|\mathcal{D}| = K$, these approaches
yield algorithms with time complexity
$\bigO(M + Nm)$, $\bigO(M + NK)$, or $\bigO(M + NM^{\frac{1}{2}})$,
where $m = \max_{D \in \mathcal{D}} |D|$ and $M = \sum_{D \in \mathcal{D}} |D|$.

An algorithm running in $\widetilde{\bigO}(M + NM^{1/3})$ time is
presented in~\cite{fastwordbreak}. In the same paper, authors also
prove that, assuming the \emph{Combinatorial $k$-Clique
  Conjecture} (see \cref{sec:lower-bound}), there is no combinatorial
algorithm for the \probname{Word Break} problem that runs in $\bigO(M +
NM^{1/3-\epsilon})$ time, for any $\epsilon > 0$.

\vspace{-1.0ex}
\paragraph{A Folklore Algorithm}

Our algorithm in \cref{sec:compressed} uses 
a subroutine a modification of the following simple
algorithm running in $\bigO(M + Nm)$ time.

\vspace{1.5ex}
\begin{algorithm}[H]
  \caption{\texttt{WordBreak}$(w,\mathcal{D})$}
  Construct a trie $\mathcal{T}$ of strings in $\mathcal{D}$\\
  $F[0] \gets 1$\\
  \For{$i \gets 0$ \KwTo $N-1$}{
    \If{$F[i]=1$}{
      \For{$j \gets i+1$ \KwTo $\min(N,i+m)$}{\label{alg:check-D-1}
        \If{$w[i+1 \dd j]\in\mathcal{D}$}{\label{alg:check-D-2}
          $F[j] \gets 1$\\
        }
      }
    }
  }
  \Return $F[N]$
\end{algorithm}
\vspace{2ex}

To implement lines~\ref{alg:check-D-1}-\ref{alg:check-D-2} in $\bigO(m)$
time, we traverse the path from the root of $\mathcal{T}$, following
edges labeled with symbols of $w[i+1 \dd \min(N,i+m)]$.
The construction of $\mathcal{T}$ takes $O(M)$ time. Hence,
the total running time is $O(M + Nm)$.
	
Note that the algorithm computes the answer for every prefix of $w$,
i.e., for every $i \in [0 \dd N]$,
$F[i] = 1$ if and only if $w[1 \dd i]$ can be partitioned into
words from $\mathcal{D}$.

\section{Compressed Word Break}\label{sec:compressed}

For the duration of this section, we fix an SLP $G = (V,\Sigma,R,S)$
and a dictionary $\mathcal{D}$.
We denote the string generated by $G$ by $w$, i.e., $L(G) = \{w\}$,
and we denote $N = |w|$. As before, we also denote
$m = \max_{D \in \mathcal{D}} |D|$ and $M = \sum_{D \in \mathcal{D}}|D|$.

\begin{definition}\label{def:M}
  Let $v \in V$ be a nonterminal. We define a two-dimensional boolean
  matrix $M_v$ of size $(m+1) \times (m+1)$ such that, for every $0
  \leq i,j \leq m$, $M_v[i,j] = 1$ holds if and only if $i+j \leq
  \len{v}$ and the string $\exp{v}[i+1 \dd \len{v} - j]$ can be
  partitioned into words from dictionary $\mathcal{D}$.
\end{definition}

\begin{remark}\label{rm:M}
  Note that if $i+j = \len{v}$, then $M_v[i,j] = 1$, i.e., the empty
  string is assumed to be a decomposition of words from a dictionary
  $\mathcal{D}$.
\end{remark}

We will now show how to compute $M_v$ in a bottom-up manner.

\begin{definition}\label{def:T}
  For every $v \in V$ such that $\rhs{v} = ab$, where $a,b
  \in V$, we define a boolean matrix $T_{v}$ of size $(m+1) \times (m+1)$
  such that, for every $0 \leq i,j \leq m$, $T_v[i,j] = 1$
  holds if and only if $i \leq \len{a}$, $j \leq \len{b}$, and
  $\exp{a}[\len{a}-i+1 \dd \len{a}] \cdot \exp{b}[1 \dd j] \in
  \mathcal{D}$.
\end{definition}

\begin{lemma}\label{lm:T}
  $T_v$ can be computed in $\bigO(m^2)$ time.
\end{lemma}
\begin{proof}
  The procedure is a simple modification of the folklore algorithm
  presented in \cref{sec:uncompressed}.
  We first construct a trie of $\mathcal{D}$. For every $i$
  from $0$ to $\min(\len{a},m)$, we then perform a traversal from the root,
  following the path spelling $\exp{a}[\len{a}-i+1 \dd \len{a}] \cdot
  \exp{b}[1 \dd \min(\len{b},m)]$. During this traversal, we check
  at what points we encountered a word in $\mathcal{D}$, and mark
  those positions in $T_{v}$.
\end{proof}

\begin{lemma}\label{lm:matmul}
  Let $v \in V$ be such that $\rhs{v} = ab$, where $a,b \in V$. Then,
  $M_v$ can be written as the product of three boolean matrices. More
  precisely, $M_v = M_a \times T_v \times M_b$ over the $(\lor,\land)$
  semi-ring.
\end{lemma}
\begin{proof}
  Consider any $i,j \in [0 \dd m]$. We will show that $M_v[i,j] = 1$
  holds if and only if $(M_a \times T_v \times M_b)[i,j] = 1$.

  ($\Rightarrow$) Let us first assume that $M_v[i,j] = 1$. Then,
  it holds $i + j \leq \len{v}$ and
  $\exp{v}[i+1 \dd \len{v} - j]$ can be partitioned into words from
  $\mathcal{D}$. In this partitioning there exists a string from
  $\mathcal{D}$ that overlaps both $\exp{a}$ and $\exp{b}$, i.e., for
  some $k,\ell \in [0 \dd m]$, the concatenation of the last $k$
  characters in $\exp{a}$ and the first $\ell$ characters in $\exp{b}$
  is a member of $\mathcal{D}$. Thus, $T_v[k,\ell] = 1$. Moreover,
  $\exp{a}[i+1 \dd \len{a} - k]$ and $\exp{b}[\ell+1 \dd \len{b} - j]$
  can then be partitioned into words from $\mathcal{D}$. This implies
  that $M_a[i,k] = 1$ and $M_b[\ell,j] = 1$. We have thus proved that
  $(M_a \times T_v \times M_b)[i,j] = 1$.

  ($\Leftarrow$) Assume that $(M_a \times T_v \times M_b)[i,j] = 1
  $. Then, there exist $k, \ell \in [0 \dd m]$, such that $M_a[i,k] =
  1$, $T_v[k,\ell] = 1$, and $M_b[\ell,j] = 1$. By \cref{def:M}, the
  substrings $\exp{a}[i+1 \dd \len{a} - k]$ and $\exp{b}[\ell+1 \dd
  \len{b} - j]$ can be partitioned into words from $\mathcal{D}$. On
  the other hand, $T_v[k,\ell] = 1$ implies that $\exp{a}[\len{a} - k
  + 1 \dd \len{a}] \cdot \exp{b}[1 \dd \ell] \in \mathcal{D}$. It remains
  to note that by
  $\exp{v} = \exp{a} \cdot \exp{b}$, we have $\exp{a}[i+1
  \dd \len{a} - k] \cdot \exp{a}[\len{a} - k + 1 \dd \len{a}] \cdot
  \exp{b}[1 \dd \ell] \cdot \exp{b}[\ell+1 \dd \len{b} - j] =
  \exp{v}[i+1 \dd \len{v} - j]$. Therefore, $\exp{v}[i+1 \dd \len{v} -
  j]$ can be partitioned into words from $\mathcal{D}$, and hence
  $M_v[i,j] = 1$.
\end{proof}

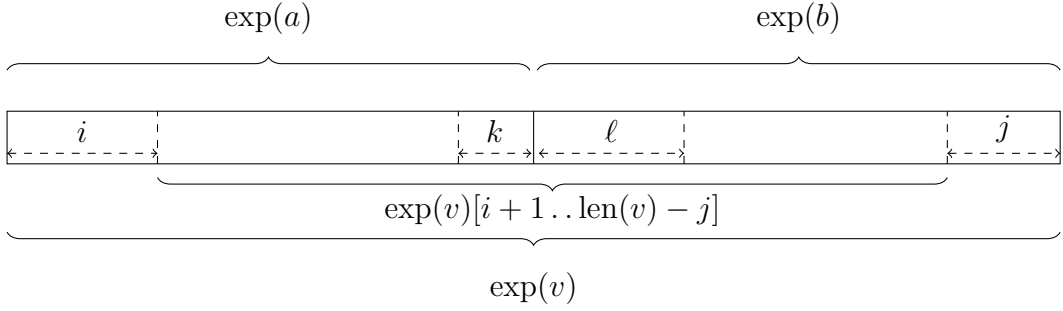
\begin{figure}
  \centering
  \begin{tikzpicture}[yscale=0.7,xscale=1]
    \draw [decorate,decoration={brace,amplitude=5pt,raise=3ex}] (0,1) -- (6.95,1) node[midway,yshift=3em]{$\exp{a}$};
    \draw [decorate,decoration={brace,amplitude=5pt,raise=3ex}] (7.05,1) -- (14,1) node[midway,yshift=3em]{$\exp{b}$};
    \draw (0,0) rectangle (14,1);
    \draw (7,0) -- (7,1);
    \draw[dashed] (2,0)  -- (2,1);
    \draw[dashed] (12.5,0) -- (12.5,1);
    \draw[dashed] (6,0)  -- (6,1);
    \draw[dashed] (9,0)  -- (9,1);
    \draw[dashed,<->] (0,0.2)  -- node[above]{$i$} (2,0.2);
    \draw[dashed,<->] (6,0.2)  -- node[above]{$k$} (6.97,0.2);
    \draw[dashed,<->] (7.07,0.2)  -- node[above]{$\ell$} (9,0.2);
    \draw[dashed,<->] (12.5,0.2) -- node[above]{$j$} (14,0.2);
    \draw [decorate,decoration={brace,amplitude=5pt,mirror,raise=1ex}]
      (2,0) -- (12.5,0) node[midway,yshift=-1.5em]{$\exp{v}[i+1 \dd \len{v} - j]$};
    \draw [decorate,decoration={brace,amplitude=5pt,mirror,raise=5ex}]
      (0,0) -- (14,0) node[midway,yshift=-4em]{$\exp{v}$};
  \end{tikzpicture}
  \caption{Illustration of the proof of \cref{lm:matmul}: $M_v[i,j] =
    1$ holds if and only if there exist $k,\ell$ such that $M_{a}[i,k]
    = 1$, $T_{v}[k,\ell] = 1$, and $M_{b}[\ell,j] = 1$.}\label{fig:matmul}
\end{figure}

\begin{theorem}\label{th:index}
  An SLP $G = (V,\Sigma,R,S)$ that generates a string $w \in \Sigma^N$
  along with a dictionary $\mathcal{D} = \{d_1,d_2,\ldots,d_K\}$ with
  $\sum_{i=1}^{K}|d_i| = M$ and $\max_{i=1}^{K}|d_i| = m$ can be
  preprocessed in $\bigO(|G|m^{\omega} + M)$ time to enable solving
  \probname{Word Break} for substrings of $w$ in $\bigO(m^2 \log N)$
  time.
\end{theorem}
\begin{proof}
  The preprocessing proceeds as follows:
  \begin{enumerate}[leftmargin=1cm]
  \item In $\bigO(|G|)$ time we apply the SLP balancing from
    \cite{balancingslp} to ensure that $G$ has height $\bigO(\log N)$.
    This increases the size of $G$ only by a constant factor.
  \item We preprocess the dictionary $\mathcal{D}$ by inserting every
    word into a trie to enable efficient membership queries. This
    takes $\bigO(M)$ time.
  \item For every $v \in V$, we compute and store the matrix $T_v$
    using \cref{lm:T}. In total, this takes $\bigO(|G|m^2)$
    time.
  \item For every $v \in V$, we compute the matrix $M_v$. To this end,
    we process all $v$ in an increasing order of size (or any
    topological order of the corresponding SLP DAG). During this
    computation, we consider two cases:
    \begin{itemize}
    \item If $\rhs{v} = a$, where $a \in \Sigma$, we set $M_v[i,j] = 0$
      for every $i+j \neq 1$ and $M_v[i,j] = 1$ for $i+j=1$.
      Furthermore, if $a \in \mathcal{D}$, we also set $M_v[0,0] = 1$.
    \item Otherwise, we apply \cref{lm:matmul} to compute $M_v$ in
      $\bigO(m^{\omega})$ time
    \end{itemize}
    Over all $v \in V$, the computation of $M_v$ takes
    $\bigO(|G|m^{\omega})$ time.
  \end{enumerate}
  In total, the preprocessing takes $\bigO(M + |G|m^{\omega})$ time.
  
  To solve \probname{Word Break} for any substring $w[i \dd j]$, we
  first compute a sequence of $p = \bigO(\log N)$ nonterminals
  $v_{1},\dots,v_{p}$ that satisfy $\exp{v_1} \cdot
  \exp{v_2} \cdot \ldots \cdot \exp{v_p} = w[i \dd j]$. During
  this decomposition, we additionally compute a sequence $u_1,
  \dots, u_{p-1}$ such that, for every $t \in [1 \dd p-1]$,
  $u_{t}$ is some common ancestor of nodes $v_{t}$ and
  $v_{t+1}$. Then, to solve \probname{Word Break} for $w[i \dd
  j]$, we evaluate the following product:
  \[
    [1,0,0,\cdots,0] \times M_{v_1} \times T_{u_1} \times
    M_{v_2} \times \cdots \times T_{u_{p-1}} \times
    M_{v_p} \times [1,0,\cdots,0]^T,
  \]
  where the two boundary vectors have dimension $m+1$.  The
  multiplications can be carried out in $\bigO(m^2)$ time by
  multiplying matrices with vectors in a left-to-right or
  right-to-left order. This yields a total running time of $\bigO(m^2
  \log N)$ per query.
\end{proof}

The above construction immediately implies the following result.

\begin{corollary}\label{cor:word-break}
  Given a dictionary $\mathcal{D} = \{d_1,d_2,\ldots,d_K\}$ and an SLP
  $G = (V,\Sigma,R,S)$ generating a string $w \in \Sigma^N$, the
  \probname{Word Break} on $(w,\mathcal{D})$ can be solved in $\bigO(M
  + |G|m^{\omega})$ time, where $M = \sum_{i=1}^{K}|d_i|$ and
  $m = \max_{i=1}^{K}|d_i|$.
\end{corollary}

\section{Lower Bound}\label{sec:lower-bound}

\paragraph{Overview}

In this section, we prove a conditional lower bound for the compressed
\probname{Word Break} problem. We will use the following
conjecture.

\begin{conjecture}[Combinatorial $k$-Clique Conjecture]\label{kclique-conjecture}
  Consider any constant $k \geq 1$. There is no combinatorial
  algorithm that, given an $n$-vertex undirected graph $\mathcal{G}$, determines
  if $\mathcal{G}$ contains a $4k$-clique in $\bigO(n^{4k-\epsilon})$ time, where
  $\epsilon > 0$.
\end{conjecture}

Consider any constant $k \geq 1$ and an undirected graph $\mathcal{G} = (V,E)$ with $V = \{1, \ldots, n\}$.
We assume that $k$, $\mathcal{G}$, $n$, and $N = n^k$ are fixed throughout this section.
The main idea of our reduction is to show that for every such $\mathcal{G}$,
there exists a string and a dictionary, denoted $w(\mathcal{G})$ and $\mathcal{D}(\mathcal{G})$,
such that $w(\mathcal{G})$ is highly compressible (and this compressed representation can be computed efficiently),
and $\mathcal{G}$ contains a $4k$-clique if and
only if $w(\mathcal{G})$ can be factorized into strings from $\mathcal{D}(\mathcal{G})$.

The structure of our reduction follows the lower bound for the
uncompressed variant of \probname{Word Break} presented
in~\cite{fastwordbreak}. Our reduction, however, differs in one major
way: The reduction in~\cite{fastwordbreak} constructs the instance of
\probname{Word Break} that enumerates all-but-two vertices in the clique, and
uses gadgets to add the two missing vertices. In our reduction, we
enumerate all $2k$-cliques, and then use gadgets to create two more
$k$-cliques that are connected to all other vertices. This changes the
construction and parameters of the output instance. Specifically, the
size of $\mathcal{D}(\mathcal{G})$ in our reduction is much larger, and depends on
$k$. In return, the output string $w(\mathcal{G})$ is highly compressible, which is
the property we need in our setting.

\paragraph{Local bi-clique gadget}

Assume we are given a $2k$-clique $A = \{a_1,a_2,\dots,a_{2k}\}$ and
we need to test if there exists a $k$-clique $B =
\{b_1,b_2,\dots,b_k\}$ in $\mathcal{G}$ such that $(A,B)$ is a bi-clique, i.e.,
$A \cap B = \emptyset$, and $a_i$ is adjacent to $b_j$ for all
$(i,j) \in [1 \dd 2k] \times [1 \dd k]$.

In our reduction, we identify (and use interchangeably) $k$-cliques in $\mathcal{G}$
with integers in $[1 \dd N]$. The string corresponding to the
$2k$-clique $A$ is:
\[
  w_A = \dol 1 2 \cdots N \hash a_1 \hash
             1 2 \cdots N \hash a_2 \hash \cdots
                          \hash a_{2k} \hash
             1 2 \cdots N \dol.
\]
For every $k$-clique $B$ and all vertices $i \in \{1, \dots, n\} \setminus A$
such that $i$ is adjacent to every node in $B$, we will add the string
$B (B+1) \cdots N \hash i \hash 1 2 \cdots (B-1)$ to our dictionary
where $B$ is also the unique integer corresponding to that
$k$-clique. We also add the strings $\dol1 2 \cdots (B-1)$ and $B
(B+1) \cdots N \dol$ for every $k$-clique $B$ to the dictionary.

If we wish to break $w_A$ into words from dictionary and assume that
the first string is $\dol 1 2 \cdots (B-1)$, then the only
way to continue is with a series of dictionary words of the form $B
(B+1) \cdots N \hash a_i \hash 1 2 \cdots (B-1)$. This implies that
the only possible candidates are the $k$-cliques $B$ adjacent to some
vertex $a_i$ in $A$. The ending is enforced to be the dictionary word
$B (B+1) \cdots N \dol$. Therefore, the string $w_A$ can be
partitioned into dictionary words if and only if there exists some $B$
such that $A \cup B$ is a $3k$-clique.

\paragraph{Local tri-clique gadget}

For the next step, we need to construct a way to extend a $2k$-clique
into a $4k$-clique. We start with a $2k$-clique $A =
\{a_1,a_2,\dots,a_{2k}\}$ and the string $w_A$ defined in the
previous paragraph.  We then construct the following string: $w_A' =
w_A \gamma w_A$. Next, for every pair of $k$-cliques $B,C \in [1 \dd
  N]$ such that $(B,C)$ is a bi-clique, along with all the dictionary
words defined in the previous paragraph, we also add the string $B
(B+1) \cdots N \dol \gamma \dol 1 2 \cdots (C-1)$ to the dictionary.

To break $w_A'$ into dictionary words, we need to find a
$k$-clique $B$ to partition the first half of $w_A'$ into dictionary
words, i.e., until $\hash a_{2k} \hash 1 2 \cdots (B-1)$. Then, we
need to find a $k$-clique $C$ such that $(B,C)$ is a
bi-clique. Therefore, we can continue only with the string $B (B+1)
\cdots N \dol \gamma \dol 1 2 \cdots (C-1)$. Finally, to partition the
remaining suffix, every vertex of $C$ must also be adjacent to every
node in $A$ due to the same reasoning as above.
Therefore, $(A,B,C)$ is a tri-clique of size $4k$ if and
only if $w_A'$ can be partitioned into substrings each belonging to
the dictionary.

\paragraph{Putting everything together}

We now combine the above gadgets for every $2k$-clique in the input
graph. To this end, we proceed similarly as in~\cite{fastwordbreak}, i.e.,
we replace every symbol in the gadgets above with 3
symbols. The purpose of this is to create three ``states'' during
the partition of the final word. These states correspond to the value
$i \bmod 3$, where $i$ is a position of the
last symbol of the currently parsed prefix.
Initially we will partition the string into elements of dictionary 
an offset of zero. Once we find the correct $2k$-clique, we will continue
matching with offset 1. Finally, we finish the correct $2k$-clique
gadget with an offset 2, and the partition of the suffix of the input
text finishes with offset 2. Formally, we proceed as follows.

\begin{definition}\label{def:w-and-D}
  The string $w(\mathcal{G})$ and a dictionary $\mathcal{D}(\mathcal{G})$ are defined as follows:

\begin{enumerate}[leftmargin=1cm]
  \item	For any string $s = s_1s_2 \cdots s_{\ell}$, let
    \begin{align*}
      f_0(s)
        &= \alpha s_1 \beta \alpha s_2 \beta \cdots \alpha s_{\ell} \beta,\\
      f_1(s)
        &= s_1 \beta \alpha s_2 \beta \alpha \cdots \alpha s_{\ell} \beta \alpha.
    \end{align*}
    For every $k$-clique $B \in [1 \dd N]$, add the following two
    strings to dictionary $\mathcal{D}(\mathcal{G})$: $f_1(\dol 1 2 \cdots (B-1))$
    and $f_1(B (B+1) \cdots N)$.
  \item For every $k$-clique $B$ and every $i \in \{1, \ldots, n\} \setminus B$ that is
    adjacent to all nodes in $B$, add the string
    $f_1(B(B+1) \cdots N \hash i \hash 1 2 3 \cdots (B-2)(B-1))$ to $\mathcal{D}(\mathcal{G})$.
  \item For every pair of $k$-cliques $B$ and $C$ such that $(B,C)$ is
    a bi-clique, add the following string to
    $\mathcal{D}(\mathcal{G})$:
    \[
      f_1(B(B+1) \cdots N \dol \gamma \dol 1 2 3 \cdots (C-2)(C-1)).
    \]
  \item For every $\sigma \in
    \{1,2,\cdots,N,\dol,\hash,\gamma,\mu\}$, add the strings
    $\alpha \sigma \beta$ and $\beta \alpha \sigma$ to
    $\mathcal{D}(\mathcal{G})$.
  \item The last set of strings added to $\mathcal{D}(\mathcal{G})$ is: $\alpha
    \mu \beta \alpha$, $\dol \beta \alpha \mu$, and $\beta \mu
    \mu$.
  \item For a $2k$-clique $A =
    \{a_1,a_2,\dots,a_{2k}\}$, let $w_A$ be the gadget defined
    earlier, i.e.,
    \[
      w_A =
        \dol 1 2 \cdots N \hash a_1 \hash
             1 2 \cdots N \hash a_2 \hash \cdots
                          \hash a_{2k} \hash
             1 2 \cdots N \dol.
    \]
    We then define
    \[
      w(\mathcal{G}) = \bigl(\textstyle\bigodot_{A}f_0(\mu w_A \gamma w_A \mu) \bigr)\mu \mu,
    \]
    where the concatenation is over all $2k$-cliques $A$.
\end{enumerate}
\end{definition}

\begin{lemma}\label{lm:equivalence}
  The string $w(\mathcal{G})$ can be partitioned into words from $\mathcal{D}(\mathcal{G})$ if and
  only if $\mathcal{G}$ contains a $4k$-clique.
\end{lemma}
\begin{proof}
  ($\Leftarrow$) Assume that $\mathcal{G}$ contains a $4k$-clique.
  Let $A \cup B \cup C$ be a partition of the vertices in this clique,
  where $|A| = 2k$ and $|B|=|C|=k$. By \cref{def:w-and-D},
  $w(\mathcal{G})$ contains $w' = f_0(\mu w_A
  \gamma w_A \mu)$ as a substring.
  We begin by partitioning the prefix of $w(\mathcal{G})$ preceding the
  substring $w'$ using words of the form $\alpha \sigma \beta$.
  The word $\alpha \mu \beta \alpha$ is then used to change the offset
  to one and match the prefix of $w'$. The next element of the partition is $f_1(\dol
  1 2 3 \cdots (B-2)(B-1))$. For every $a \in A$, we then match the next
  $2k$ substrings of $w'$ using $f_1(B(B+1)\cdots N \hash a \hash 1 2 3
  \cdots (B-2)(B-1))$. Next, we include the $\gamma$ symbol in $w'$ by
  using the word $f_1(B(B+1) \cdots N \dol \gamma \dol 1 2 3
  \cdots(C-2)(C-1))$. The second part of $w'$ is completed by
  repeatedly using $f_1(C(C+1)\cdots N \hash a \hash 1 2 \cdots
  (C-1))$. We finish partitioning $w'$ by using $f_1(C(C+1) \cdots N)$ followed by
  $\dol \beta \alpha \mu$. Now the offset is two and we can complete
  the partitioning using strings
  $\beta \alpha \sigma$, followed
  by $\beta \mu \mu$.
	
  ($\Rightarrow$) Assume that $w(\mathcal{G})$ can be broken into words of
  $\mathcal{D}(\mathcal{G})$. The last word used is $\beta \mu \mu$ preceded by some
  number of words of the form $\beta \alpha \sigma$. Following the partition
  right to left, the very first word not of the form $\beta \alpha
  \sigma$ must be $\dol \beta \alpha \mu$. After including this string
  in the partition, the remaining prefix of
  $w(\mathcal{G})$ ends with the string $f_0(\mu w_A \gamma w_A \mu)$, with
  the last five symbols removed, and $A$ is an
  integer corresponding to some $2k$-clique. The offset has now
  changed from two to one. The preceding word here must be $f_1(C(C+1)
  \cdots N)$ where $C$ is an integer corresponding to some
  $k$-clique. The next $2k$ words must of the form
  $f_1(C(C+1) \cdots N \hash a \hash 1 2 \cdots (C-2)(C-1))$. By
  Step~2 of the procedure above, $a$ is adjacent to every vertex in
  $C$. Thus, $(A,C)$ forms
  a bi-clique. Now, the only way to match $\gamma$ is using the word
  $f_1(B(B+1) \cdots N \dol \gamma \dol 1 2 \cdots (C-1))$ for some
  integer $B$ corresponding to a $k$-clique. By Step~3, $(B,C)$ forms
  a bi-clique. The next $2k$ words are
  of the form $f_1(B(B+1) \cdots N \hash a \hash 1 2 \cdots
  (B-2)(B-1))$, where $a$ is a vertex in the $2k$
  clique $A$. Again, by Step~2, every such $a$ is
  adjacent to all of $B$. Thus, $(A,B)$ forms a bi-clique. Since each
  of the pairs $(A,B),(B,C)$ and $(A,C)$ are bi-cliques, we
  conclude that $(A \cup B \cup C)$ forms a $4k$-clique. The remaining
  prefix of $w(\mathcal{G})$ must be matched by an $\alpha \mu \beta \alpha$
  (switching the offset from 1 to 0), and then preceded by strings
  of the form $\alpha \sigma \beta$ (which all end at offset 0).
\end{proof}

\begin{lemma}\label{lm:D-properties}
  It holds $\sum_{D \in \mathcal{D}(\mathcal{G})}|D| = \bigO(N^3)$
  and $\max_{D \in \mathcal{D}(\mathcal{G})} |D| = \bigO(N)$.
\end{lemma}
\begin{proof}
  Recall the construction presented in \cref{def:w-and-D}.
  Step~1 adds $\bigO(N)$ strings each of length
  $\bigO(N)$. The second step inserts $\bigO(N\cdot k)$ words into
  $\mathcal{D}(\mathcal{G})$, each of length $\bigO(N)$. Step~3
  adds $\bigO(N^2)$ strings to $\mathcal{D}(\mathcal{G})$, each of length
  $\bigO(N)$. Step~4 adds
  $\bigO(N)$ strings, each of length $3$, to $\mathcal{D}(\mathcal{G})$. Step~5
  contributes $\bigO(1)$ to the total length. We thus obtain
  $\sum_{D \in \mathcal{D}(\mathcal{G})}|D| = \bigO(N^3)$
  and $\max_{D \in \mathcal{D}(\mathcal{G})} |D| = \bigO(N)$.
\end{proof}

\begin{lemma}\label{lm:slp}
  There exists an SLP $G$ such that $L(G) = \{w(\mathcal{G})\}$
  and $|G| = \bigO(N^2)$. Moreover, given $\mathcal{G}$, we
  can construct $G$ in $\bigO(N^2)$ time.
\end{lemma}
\begin{proof}

  Let $U$ be a nonterminal with $\exp{U} = 1 2 \cdots N$. For any
  $2k$-clique $A = \{a_1,a_2,\dots,a_{2k}\}$, $w_A = \dol \exp{U}
  \hash a_1 \hash \exp{U} \hash a_2 \hash \dots \hash a_{2k} \hash
  \exp{U} \dol$. Thus, $w_A$ can be described using $\bigO(k)$
  nonterminals. Furthermore, $f_0(\mu w_A \gamma w_A \mu)$ can also be
  described using $\bigO(k)$ nonterminals since all $f_0(\cdot)$ does is add
  alternating $\beta \alpha$ between every letter of a string. Since
  there are $\bigO(N^2)$ different $2k$-cliques in $\mathcal{G}$,
  $w(\mathcal{G})$ can be described as a concatenation of $\bigO(k \cdot N^2)$
  nonterminals. By converting the grammar into an SLP, we
  conclude that $w(\mathcal{G})$ is generated by an SLP of size $\bigO(N^2)$.

  The above SLP is easily constructed in
  $\bigO(N^2 \cdot k^2) = \bigO(N^2)$ time.
\end{proof}

\begin{lemma}\label{lm:construction}
  Given $\mathcal{G}$, we can construct $\mathcal{D}(\mathcal{G})$ in
  $\bigO(N^3)$ time.
\end{lemma}
\begin{proof}
  In the first step of \cref{def:w-and-D}, for every $k$-clique $B$,
  we add the string of length $n^k$ to
  $\mathcal{D}(\mathcal{G})$. Enumerating all $k$-cliques in an
  $n$-vertex graph takes $\bigO(n^k \cdot k^2)$ time.  Thus, the first
  step in total takes $\bigO(n^k \cdot k^2 + n^{2k})$ time.  By a
  similar argument, the second step takes $\bigO(n^{k+1} \cdot k^2 +
  n^{2k+1})$, and the third step takes $\bigO(n^{2k} \cdot k^2 +
  n^{3k})$ time. Step~4 takes $\bigO(n^k)$ time, and Step~5 takes
  $\bigO(1)$ time.  In total, the construction of
  $\mathcal{D}(\mathcal{G})$ takes $\bigO(n^{3k} + n^{2k} \cdot k^2) =
  \bigO(n^{3k}) = \bigO(N^3)$ time (recall, that $k = \bigO(1)$).
\end{proof}

\begin{theorem}\label{th:lower-bound}
  Consider a string $u$ and let $G = (V,\Sigma,R,S)$ be an SLP with
  $L(G) = \{u\}$.  Let $\mathcal{D}$ be a dictionary, and let $M =
  \sum_{D \in \mathcal{D}} |D|$ and $m = \max_{D \in \mathcal{D}}
  |D|$.  Assuming the Combinatorial $k$-Clique Conjecture, there is no
  combinatorial algorithm that, given $G$ and $\mathcal{D}$, solves
  \probname{Word Break} for $(u,\mathcal{D})$ in $\bigO(|G| \cdot m^{2
  - \epsilon} + M)$ time, where $\epsilon > 0$.
\end{theorem}
\begin{proof}
  Suppose that this is not true.
  We will show that this implies that we can check if $\mathcal{G}$
  contains a $4k$-clique in $\bigO(n^{4k-\epsilon'})$ time for some
  constant $\epsilon'>0$ (i.e., \cref{kclique-conjecture} is not
  true).
  The algorithm works as follows:
  \begin{enumerate}[leftmargin=1cm]
  \item Using \cref{lm:slp,lm:construction}, in $\bigO(N^3)$ time we
    construct $\mathcal{D}(\mathcal{G})$ and an SLP $G$ of size $|G| =
    \bigO(N^2)$ satisfying $L(G) = \{w(\mathcal{G})\}$.
  \item Using the hypothetical algorithm, we check if
    $w(\mathcal{G})$ can be partitioned into words from
    $\mathcal{D}(\mathcal{G})$. To bound its runtime, recall that
    by \cref{lm:D-properties}, the values
    $M_{\mathcal{D}} = \sum_{D \in \mathcal{D}(\mathcal{G})}|D|$ and
    $m_{\mathcal{D}} = \max_{D \in \mathcal{D}(\mathcal{G})}|D|$ satisfy
    $M_{\mathcal{D}} = O(N^3)$ and $m_{\mathcal{D}} = \bigO(N)$.
    Thus, the algorithm takes 
    $\bigO(|G| \cdot m_{\mathcal{D}}^{2-\epsilon} + M_{\mathcal{D}}) =
    \bigO(N^2 \cdot N^{2-\epsilon} + N^3)$ time.
  \end{enumerate}
  In total, we spend $\bigO(N^{4-\epsilon} + N^3) =
  \bigO(n^{4k - 4\epsilon} + n^{3k})$ time
  to determine if $w(\mathcal{G})$ can be partitioned into words from
  $\mathcal{D}(\mathcal{G})$. By \cref{lm:equivalence}, this is equivalent to
  checking if $\mathcal{G}$ contains a $4k$-clique. 
  Thus, we obtain the above claim with constant
  $\epsilon' = 4\epsilon$.
\end{proof}

\Section{References}
\bibliographystyle{IEEEbib}
\bibliography{paper}

\end{document}